\documentclass[12pt]{article}%
\usepackage[letterpaper, top=1in, bottom=1in, left=1in, right=1in]{geometry}
\usepackage{amsfonts, amsmath, amsthm, amssymb}
\usepackage{enumitem, graphicx, stmaryrd, color, bbm}
\usepackage{url, algorithm}
\usepackage{float}
\floatstyle{ruled}
\usepackage[T1]{fontenc}
\usepackage{footmisc}

\usepackage[pagebackref=true,colorlinks]{hyperref}
	\hypersetup{linkcolor=[rgb]{.7,0,.7}}
	\hypersetup{citecolor=[rgb]{.5,0,.5}}
	\hypersetup{urlcolor=[rgb]{.7,0,.7}}

\newtheorem{theorem}{Theorem}[section]

\newtheorem{definition}[theorem]{Definition}

\newtheorem{corollary}[theorem]{Corollary}
\newtheorem{remark}[theorem]{Remark}
\newtheorem{proposition}[theorem]{Proposition}

\newtheorem{question}[theorem]{Question}
\newtheorem{fact}[theorem]{Fact}

\renewcommand{\implies}{\Longrightarrow}		
\newcommand{\abs}[1]{\left|#1\right|}		
\newcommand{\N}{\mathbb{N}} 			  		
\newcommand{\F}{\mathbb{F}}					

\newcommand{\mc}{\mathcal}

\newcommand{\set}[1]{\left\{ #1 \right\}}   
\newcommand{\brac}[1]{\left( #1 \right)}    
\newcommand{\sqbrac}[1]{\left[ #1 \right]}  

\newcommand{\norm}[1]{\left\vert #1 \right\vert}			

\newcommand{\val}{\textnormal{val}} 	 			
\newcommand{\up}[1]{{\left(#1\right)}}

\begin{document}

\title{Multiplayer Parallel Repetition Is the Same as High-Dimensional Extremal Combinatorics}
\author{Kunal Mittal\thanks{New York University. E-mail: \href{kunal.mittal@nyu.edu}{\texttt{kunal.mittal@nyu.edu}.} Research supported by a Simons Investigator Award, and NSF Award CCF-2007462. This work also appeared in the author's PhD thesis~\cite{Mit25}.}}
\date{}
\maketitle

\begin{abstract}
	We show equivalences between several high-dimensional problems in extremal combinatorics and parallel repetition of multiplayer (multiprover) games over large answer alphabets.
	This extends the forbidden-subgraph technique, previously studied by Verbitsky (Theoretical Computer Science 1996), Feige and Verbitsy (Combinatorica 2002), and H{\k a}z{\l}a, Holenstein and Rao (2016), to all $k$-player games, and establishes new connections to problems in combinatorics.
	We believe that these connections may help future progress in both fields.
\end{abstract}


\newpage

\section{Introduction}

We study connections between two seemingly unrelated topics: parallel repetition of multiplayer games, and high dimensional problems in extremal and additive combinatorics.

\subsection{Parallel Repetition of Multiplayer Games}

In a $k$-player game $\mc G$, a referee samples a tuple $(x^\up{1},\dots,x^\up{k})$ of questions from some joint distribution $\mu$.
Then, for each $j\in [k]$, question $x^\up{j}$ is sent to the $j$\textsuperscript{th} player, to which they respond back with an answer $a^\up{j}$ (that depends only on $x^\up{j}$).
The referee then declares whether the players win or lose based on the evaluation of a predicate $V(x^\up{1},\dots,x^\up{k})$.
The value of the game $\mc G$, denoted $\val(\mc G)$, is the maximum winning probability (with respect to the distribution $\mu$) over all possible strategies of the $k$ players.

Given a game $\mc G$ with $\val(\mc G)<1$, it is natural to study how the value of the game behaves under \emph{parallel repetition}~\cite{FRS94}. 
The $n$-fold {parallel repetition of the game $\mc G$, denoted $\mc G^n$, roughly speaking, is a $k$-player in which the players play $n$ copies of $\mc G$ in parallel, and are required to win all of them.
Formally, this game proceeds as follows:
for each $i\in [n]$, the referee samples questions $(x_i^\up{1},\dots, x_i^\up{k})\sim \mu$ independently.
Then, for each $j\in [k]$, the questions $(x_1^\up{j},\dots,x_n^\up{j})$ are sent to the $j$\textsuperscript{th} player, to which they respond back answers $(a_1^\up{j},\dots,a_n^\up{j})$.
The referee declares that the players win if and only if $V(x_i^\up{1},\dots,x_i^\up{k},a_i^\up{1},\dots,a_i^\up{k}) = 1$ for each $i\in [n]$.

Observe that $\val(\mc G^n)\geq \val(\mc G)^n$, as this is the value obtained when the players play optimal strategies independently in each coordinate.
Intuitively, one might expect this inequality to be tight, since the questions received by the players for different copies of the game are independent; however, this is not true.
There are games with $\val(\mc G^n) \gg \val(\mc G)^n$~\cite{For89, Fei91, FV02, Raz11}.
The key reason for this is that in the game $\mc G^n$, while the referee treats the $n$ copies as independent, the players may not; that is, each answer of a player may depend on all of the questions they receive.

The special case of 2-player games is well-understood, where Raz proved that for any game $\mc G$ with $\val(\mc G)<1$, the value $\val(\mc G^n)$ decays exponentially in $n$~\cite{Raz98}.
Subsequent works have improved the constants in this bound and even led to bounds depending on the initial game value $\val(\mc G)$~\cite{Hol09, BRRRS09, Rao11, RR12, DS14, BG15}.
These and related works have led to many applications in various areas of mathematics: in the theory of interactive proofs~\cite{BOGKW88}, PCPs and hardness of approximation~\cite{FGLSS96, ABSS97,ALMSS98,AS98, BGS98, Fei98, Has01,Kho02a, Kho02b, GHS02, DGKR05, DRS05}, geometry of foams~\cite{FKO07, KORW08, AK09, BM21}, quantum information~\cite{CHTW04}, and communication complexity~\cite{PRW97,BBCR13, BRWY13}.
The reader is referred to this survey~\cite{Raz10} for more details.

In the case of $k$-player games, for $k\geq 3$, the only known general bound says that for any game $\mc G$ with $\val(\mc G)<1$, it holds that $\val(\mc G^n) \leq \frac{1}{\alpha(n)}$, where $\alpha(n)$ is a slowing growing inverse-Ackerman function~\cite{Ver96}.
The weak bounds here follow from the black-box use of the density Hales-Jewett theorem~\cite{FK91, Pol12}.
Recent work has made progress on special classes of multiprover games~\cite{DHVY17, HR20, GHMRZ21, GHMRZ22, GMRZ22, BKM23, BBKLM24, BBKLM25}, however, the general question remains wide open.

Multiplayer parallel repetition has several applications.
It is known that a strong parallel repetition theorem for a certain class of multiplayer games implies super-linear lower bounds for non-uniform Turing machines~\cite{MR21}.
Additionally, it is believed that progress in proving multiplayer parallel repetition bounds may lead to an improved understanding of multiparty communication complexity in the number-on-forehead (NOF) model, a problem which is intimately connected to circuit lower bounds.

\subsection{Arithmetic Combinatorics}

The fields of extremal and additive combinatorics have witnessed remarkable progress over the last century.
In this section, we survey several central problems and results in this area;
for a comprehensive introduction, the reader is referred to the excellent works of Kowalski~\cite{Kow24} and Peluse~\cite{Pel24}.

\paragraph{Arithmetic Progressions.}
One of the best-known early results in the field of arithmetic combinatorics is van der Waerden's theorem~\cite{vdWar27}:

\begin{theorem}
	For every positive integers $r,k\in \N$, there exists an integer $n\in \N$, such that if the numbers $[n] = \set{1,2,\dots,n}$ are colored with $r$ colors, then at least one of the colors contains a $k$-term arithmetic progression.
	
	By a $k$-term arithmetic progression, we mean a sequence of the form \[a,\ a+d, \ a+2d,\dots, a+(k-1)d,\]
	for integers $a,d$ with $d\not=0$.
\end{theorem}

This fundamental result naturally led to a \emph{density version} of the same question:
\begin{question}
	What is the largest density $\frac{\abs{A}}{n}$ of a subset $A\subseteq [n]$ containing no $k$-term arithmetic progression?
\end{question}

The first non-trivial case is that of $k=3$, and this was studied by Roth, who showed that any set $A\subseteq [n]$ of size $\abs{A}\geq \Omega\brac{\frac{n}{\log\log n}}$ must contain a 3-term arithmetic progression~\cite{Roth53}.
The question for longer progressions remained largely open until the seminal work of Szemer\'edi~\cite{Sze75}, which made substantial progress and showed that for any $k\in \N$, any set $A\subseteq [n]$ of size $\abs{A}\geq \Omega_k(n)$ must contain a $k$-term arithmetic progression.

Szemer\'edi's proof was purely combinatorial, and led to the development of powerful graph theoretic tools like Szemer\'edi's regularity lemma;
however, it gave extremely weak quantitative bounds.
Subsequent decades saw significant efforts to obtain better quantitative bounds, mostly via analytic means.
A major advancement was made by Gowers, who observed that Fourier analysis, used by Roth in his proof, was not sufficient to address progressions of length $k\geq 4$.
This insight led him to develop higher-order Fourier analysis, thereby providing reasonable bounds for progressions of all lengths~\cite{Gow98, Gow01}.

While the past few decades have seen numerous advances~\cite{HB87, Sze90, Bou99, Bou08, San11, San12, BK12, Blo16, Sch21, BS21}, here we only mention some recent breakthrough results that dramatically pushed the above bounds.
Kelley and Meka proved that sets free of 3-term progressions are of density at most $2^{-\brac{\log n}^{\Omega(1)}}$~\cite{KM23}.
This comes tantalizingly close to the almost matching construction of Behrend of size $2^{-O(\sqrt{\log n})}$~\cite{Beh46}.
For sets free of $k$-term progressions, Leng, Sah, and Sawhney prove a density upper bound of $2^{-\brac{\log\log n}^{\Omega_k(1)}}$~\cite{LSS24}.

The high-dimensional version of this problem, with the set $[n]$ replaced by $\F_p^n$, has also attracted a lot of interest.
Such finite field variants often tend to be more tractable than their integer counterparts, and serve as a useful proxy to develop tools to attack the integer problems.
Furthermore, such problems often have many applications in theoretical computer science.
We note that in this setting, for $k=3$, exponentially small upper bounds have been achieved via the polynomial method~\cite{CLP17, EG17}: any set $A\subseteq \F_p^n$ containing no 3-term progression is of size $\abs{A}\leq c^n$, for some $c<p$.
Obtaining such strong bounds for $k\geq 4$ remains a notorious open problem.

\paragraph{Combinatorial Lines.} The Hales-Jewett Theorem~\cite{HJ63} represents a profound generalization over the above results, shifting from arithmetic progressions to more general combinatorial structures, defined as follows:

\begin{definition}\label{defn:comb_line} (Combinatorial lines)
	Let $q, n\in \N$.
	A collection of $q$ points $a(1),\dots, a(q) \in [q]^n$ is said to be a \emph{combinatorial line} if:
	\begin{enumerate}
		\item For all $i\in [n]$, either $(a(1)_i,\dots, a(q)_i)= (1,2,3,\dots,q)$ or $a(1)_i=a(2)_i=\dots=a(q)_i$.
		\item For some $i\in [n]$, it holds that $(a(1)_i,\dots, a(q)_i) = (1,2,3,\dots,q)$.
	\end{enumerate}
	We define $r_q(n) = \frac{\abs{A}}{q^n}$, where $A\subseteq [q]^n$ is a largest set  containing no combinatorial line.
\end{definition}

The Hales-Jewett Theorem says that for any $c,q\in \N$, and large enough $n$, any $c$-coloring of $[q]^n$ contains a monochromatic combinatorial line~\cite{HJ63}.
It can be shown that this generalizes the van der Waerden's theorem.

Furstenberg and Katznelson, using techniques from ergodic theory, proved a density version of the Hales-Jewett theorem~\cite{FK91}.
This says that for every $q\in \N$, $r_q(n) \to 0$ as $n\to \infty$.
The current best-known decay bounds for $r_q(n)$ are of the form $\frac{1}{\alpha(n)}$, where $\alpha(n)$ is an extremely slow growing inverse-Ackermann function~\cite{Pol12}.
We note that the special cases $q=1,2$ are well-understood, with $r_1(n)=0$ trivially, and $r_2(n) = \frac{1}{2^n}\binom{n}{\lfloor{n/2}\rfloor} = \Theta\brac{\frac{1}{\sqrt{n}}}$ by Sperner's Theorem.
A very recent work of Bhangale, Khot, Liu, and Minzer~\cite{BKLM25} shows that $r_3(n) \leq \brac{\log\log\log\log{n}}^{-\Omega(1)}$.

\paragraph{Corners and Squares.}

Related problems of interest include those of determining the maximum density of sets which are corner-free and square-free.

\begin{definition}\label{defn:corner_sq}
	Let $G$ be a set in some underlying Abelian group.\footnote{we will only be interested in the cases $G=[n]\subseteq \mathbb{Z}$ or $G=\F_2^n$}
	\begin{itemize}
		\item A \emph{corner} in $G\times G$ is a set of the form \[\set{(x,y), (x+d,y), (x,y+d)},\] where $x,y,d\in G,\ d\not=0$.
		
			Define $r_{\angle}(G)$ to be the maximum density $\frac{\abs{A}}{\abs{G}^2}$ of a set $A\subseteq G\times G$ containing no corner.
		
		\item A square in $G\times G$ is a set of the form \[\set{(x,y), (x+d,y), (x,y+d), (x+d,y+d)},\] where $x,y,d\in G,\ d\not=0$.
			
			Define $r_{\square}(G)$ to be the maximum density $\frac{\abs{A}}{\abs{G}^2}$ of a set $A\subseteq G\times G$ containing no square.
	\end{itemize}
	Note that every square contains a corner, and so $r_{\angle}(G)\leq r_{\square}(G)$.
\end{definition}

Ajtai and Szemer\'edi showed that $r_{\angle}([n]) = o_n(1)$~\cite{AS74}; this also follows from the density Hales-Jewett theorem mentioned earlier.
The best-known upper bounds here are $r_{\angle}([n]) \leq \frac{1}{\brac{\log\log{n}}^{\Omega(1)}}$~\cite{Shk06} and $r_{\angle}(\F_2^n) \leq O\brac{\frac{\log\log{n}}{\log{n}}}$~\cite{LM07}.
It remains open whether an exponential upper bound, as for the case of 3-term arithmetic progressions, can be proven in the finite field setting for this problem.

For square-free sets, to the best of our knowledge, the best-known upper bounds follow from the density Hales-Jewett theorem: $r_{\square}(\F_2^n) \leq r_4(n)$, since every combinatorial line in $\set{1,2,3,4}^n$ is a square in $\F_2^n\times \F_2^n$, by identifying $\set{1,2,3,4}$ with $\F_2\times \F_2$.
Obtaining improved bounds here is a very challenging open problem and is of central interest (Problem 4.4 in Peluse~\cite{Pel24}).

\subsection{This Work}

In this work, we prove connections between multiplayer parallel repetition bounds and bounds for several central problems in additive combinatorics.
These connections follow from studying the so-called forbidden subgraph bounds:
For a game multiplayer $\mc G$ with $\val(\mc G)<1$, very roughly speaking, a forbidden subgraph in $\mc G^{n}$ is a \emph{pattern} that cannot be present inside the winning set corresponding to any player strategies.
Upper bounds on the sizes of sets free of such patterns naturally translate to similar bounds on the value of the game $\mc G^n$, and this is the forbidden subgraph method (see Section~\ref{sec:for_sub} for details).

The forbidden subgraph method was one of the first approaches to prove parallel repetition bounds for 2-player games.
Using this technique, Cai, Condon, and Lipton showed an exponential decay bound for \emph{free games}, which are games in which the questions to the two players come from a product distribution~\cite{CCL92}.
Verbitsky proved that for any $2$-player game, combinatorial lines (of an appropriate length) form a forbidden subgraph; this extends to all $k$-player games (see Section~\ref{sec:dhj_conn}) and gives the current best decay bound for the value of parallel repetition of general $k$-player games with value less than 1~\cite{Ver96, FV02}.

Following the works~\cite{Ver96, FV02, HHR16} we show that when the game $\mc G$ is allowed to have a large answer alphabets, bounds obtained via forbidden subgraph are the best possible (Theorem~\ref{thm:game_val_lb_forbidden}).
Then, we use this to show several connections between multiplayer parallel repetition and the problems described in the previous section.

The first connection is to combinatorial lines and the density Hales-Jewett theorem. We note that this was already observed in~\cite{FV02, HHR16}:

\begin{theorem}\label{thm:intro_dhj} (Multiplayer Parallel Repetition and Combinatorial Lines)
\begin{enumerate}
	\item (Corollary~\ref{corr:val_ub_forbidden}, Proposition~\ref{prop:comb_lin_forb}) For every game $\mc G$ with $\val(\mc G)< 1$, combinatorial lines are forbidden subgraphs for the game $\mc G$. In particular, for every $n\in \N$, it holds that  $\val(\mc G^n)\leq r_k(n)$, where $k$ is a constant depending on the game $\mc G$.\footnote{We assume here that the question distribution $\mu$ of $\mc G$ is uniform over its support; see Remark~\ref{remk:unif_dist}\label{footnote:unif}}
	\item (Corollary~\ref{corr:hj_inv}) Let $q\geq 3$ be a positive integer, and let \[Q = \set{(1,0,0,\dots,0), (0,1,0,\dots,0), \dots,(0,0,\dots,0,1)} \subseteq \set{0,1}^q\] be the set of size $q$ containing all vectors with 1 one and $(q-1)$ zeros. Then, there exists a $q$-player game $\mc G$ over questions $Q$, with an infinite answer alphabet, such that for every $n\in \N$, it holds that $\val\brac{\mc G^{n}} = r_q(n)$.\footnote{We note that if the game is allowed to depend on $n$, we can assume the answer length is finite and bounded by $\exp(n)$; see Theorem~\ref{thm:game_val_lb_forbidden} and Remark~\ref{remk:inf_length}\label{footnote:inf_length}}
\end{enumerate}
\end{theorem}

The question set $Q$ above has been well-studied in the past, in the context of a specific 3-player game known as the anti-correlation game.
In the anti-correlation game, the referee samples questions $(x,y,z)\sim \set{(1,0,0), (0,1,0), (0,0,1)}$ uniformly at random, and the two players who get input 0 are required to produce different answers in $\set{0,1}$.
This game has value 2/3, and Feige showed the value of the 3-fold repetition of this game is also $2/3$~\cite{Fei95}.
Later, Holmgren and Yang present this game as an example of a game whose non-signaling value is less than one (2/3 in this case), and whose non-signaling value does not decrease at all under parallel repetition~\cite{HY19}; this is in stark contrast to 2-player games where the non-signaling value decays exponentially with the number of repetitions $n$~\cite{Hol09}.
For the anti-correlation game, exponential decay bounds on the value of parallel repetition were proven recently~\cite{GHMRZ22}, and polynomial bounds were proven for all 3-player games with the question set $Q= \set{(1,0,0), (0,1,0), (0,0,1)}$ and with constant answer lengths~\cite{GMRZ22}.

We also mention here the recent work of Bhangale, Khot, Liu, and Minzer~\cite{BKLM25}, which shows that $r_3(n) \leq \brac{\log\log\log\log{n}}^{-\Omega(1)}$, which we believe is very interesting in the context of Theorem~\ref{thm:intro_dhj}.
At a very high level, the crux of their proof is to show that any set $A\subseteq \set{1,2,3}^n=[3]^n$ free of combinatorial lines of length 3, must have increased density on the intersection of what are known as \emph{insensitive sets}.
When viewed under the lens of parallel repetition (of the corresponding game under Theorem~\ref{thm:intro_dhj}), these insensitive sets turn out to be subsets of inputs corresponding to a single player, and their intersection is hence a \emph{product set} among the 3 players.
Such product sets arise naturally when studying parallel repetition, for example in the celebrated proof of Raz~\cite{Raz98}!
We believe this connection can potentially lead to a better understanding of the insensitive sets one needs to study when looking at certain extremal combinatorics problems.

Next, we show a similar connection between square-free sets and the well-known GHZ game~\cite{GHZ89}:

\begin{theorem}\label{intro_squares} (Multiplayer Parallel Repetition and Squares; Corollary~\ref{corr:square})

Let $Q = \set{(x,y,z)\in \F_2^3: x+y+z=0}$.
Then, it holds that:

\begin{enumerate}
	\item For every 3-player game $\mc G$ over the question set $Q$, and $\val(\mc G)< 1$, squares are forbidden subgraphs for $\mc G^n$. In particular, for every $n\in \N$, it holds that  $\val(\mc G^n)\leq r_{\square}(\F_2^n)$.\footref{footnote:unif}
	\item There exists a $3$-player game $\mc G$ over questions $Q$, with an infinite answer alphabet, such that for every $n\in \N$, it holds that $\val\brac{\mc G^{n}} = r_{\square}(\F_2^n)$.\footref{footnote:inf_length}
\end{enumerate}
\end{theorem}

We note that the first point in the above theorem was already (implicitly) exploited in the work~\cite{GHMRZ21}, which provides polynomial decay bounds for games with the above question set $Q$, and small answer alphabets; see Remark~\ref{rmk:square_in_par_rep}.

Next, we generalize the above result to \emph{grids} over an arbitrary finite field $\F$.
In terms of extremal combinatorics, we note that in a sense grids seem to capture the largest linear forbidden structures.
On the other hand, in terms of parallel repetition, we show that these capture all games where the question distribution is supported over a linear subspace.

\begin{definition}\label{defn:grid_free} (Grid-free sets)
	Let $k,n\in \N$.
	
	A grid in $(\F^n)^k$ is a set of the form $\set{x + \alpha\otimes d : \alpha\in \F^k}$ for some $x\in (\F^n)^k$,\ $d\in \F^n,\ d\not=0$.
	We use the Kronecker product notation $\alpha\otimes d = (\alpha^\up{1}d,\dots,\alpha^\up{k}d) \in (\F^{n})^k,$ for $\alpha\in \F^k,\ d\in \F^n$.
	
	Define $r_{grid}(\F,k,n)$ as the maximum density $\frac{\abs{A}}{\abs{\F}^{kn}}$ of a set $A\subseteq (\F^n)^k$ containing no grid.
\end{definition}

We prove the following result:

\begin{theorem}\label{thm:intro_grid} (Multiplayer Parallel Repetition and Grids)

\begin{enumerate}
	\item (Proposition~\ref{prop:lin_game_grid_free}) Let $\mc G$ be a multiplayer game with $\val(\mc G)< 1$, whose questions are supported over an affine subspace (over $\F$) of dimension $k$.
	 Then, for every $n\in \N$, it holds that  $\val(\mc G^n)\leq r_{grid}(\F, k, n)$.\footref{footnote:unif}
	\item (Corollary~\ref{corr:grid_free}) Suppose $\abs{\F} = p^r$, for some prime $p$, and $r\in \N$.
	Then, for every integer $k\geq 2$, there exists a $(k+r)$-player game $\mc G$, whose questions are supported over an affine subspace (over $\F$) of dimension $k$, with an infinite answer alphabet, and such that for every $n\in \N$, it holds that $\val(\mc G^{n}) = r_{grid}(\F,k, n)$.\footref{footnote:inf_length}
\end{enumerate}
\end{theorem} 

Finally, we note that a general consequence of all of the above results is that any parallel repetition bounds that do not depend on the answer lengths translate to similar bounds on central problems in extremal combinatorics.
While it may be very hard to prove such bounds in a black-box manner, we believe that the connections presented in this work may lead to a transfer of tools used in the two fields.
For example, to the best of our knowledge, the techniques used in parallel repetition (for constant-sized answers)~\cite{Raz98}, like induction on the parameter $n$ and information theory, seem to be very different than the techniques generally used in extremal/additive combinatorics.

\section{Multiplayer Games}

\begin{definition}\label{defn:multiplayer_games}(Multiplayer game)
	A $k$-player game $\mc G$ is a tuple $\mc G = (\mc X, \mc A, \mu, Q, V)$, where:
	\begin{enumerate}
		\item $\mc X = \mc X^\up{1} \times\dots\times \mc X^\up{k}$ is a finite product set, with $\mc X^\up{1}, \dots, \mc X^\up{k}$ denoting the sets of possible questions for the $k$ players, respectively.
		\item $\mc A = \mc A^\up{1} \times\dots\times \mc A^\up{k}$ is a finite product set, with $\mc A^\up{1}, \dots, \mc A^\up{k}$ denoting the sets of possible answers for the $k$ players, respectively.
		\item $\mu$ is a distribution over the set $\mc X$, with support $Q\subseteq \mc X$.
		\item $V:\mc X\times \mc A\to\set{0,1}$ is a predicate.
	\end{enumerate}
\end{definition}

The game $\mc G$ proceeds as follows: A verifier samples questions $X = (X^\up{1}, \dots, X^\up{k})\sim \mu$; then, for each $j\in [k]$, the verifier sends the question $X^\up{j}\in \mc X^\up{j}$ to the $j\textsuperscript{th}$ player, to which the player responds back with answer $A^\up{j} \in \mc A^\up{j}$.
Finally, the verifier declares that the players win if and only if $V\brac{X^\up{1},\dots,X^\up{k},A^\up{1},\dots,A^\up{k}} = 1$.

The value of the game $\mc G$, then, is the maximum possible winning probability for the $k$ players, formally defined as: 

\begin{definition}(Value of a game)
	Let $\mc G = (\mc X, \mc A, \mu, Q, V)$ be a $k$-player game.
	
	For a sequence $\brac{f^\up{j}:\mc X^\up{j}\to \mc A^\up{j}}_{j\in[k]}$ of functions, let the function $f = f^\up{1}\times\dots\times f^\up{k} : \mc X \to \mc A$ be given by $f\brac{x^\up{1},\dots,x^\up{k}} = \brac{f^\up{1}(x^\up{1}), \dots, f^\up{k}(x^\up{k})}$. 
	We use the term \emph{product functions} for functions $f$ defined in this manner, and the functions $(f^\up{j})_{j\in[k]}$ are called \emph{player strategies}.
	
	The value $\val(\mc G)$ of the game $\mc G$ is defined as 
	\[\val(\mc G) = \max_{f = f^\up{1}\times \dots\times f^\up{k}}\ \Pr_{X\sim \mu}\sqbrac{V(X,f(X))=1},\]
	where the maximum is over all product functions (or player strategies).
\end{definition}

\begin{fact}\label{fact:rand_no_help}
	The value of the game is unchanged even if we allow the player strategies $\brac{f^\up{j}}_{j\in [k]}$  to be randomized; that is, we allow the strategies to depend on some additional shared and private randomness.
	This is because there always exists an optimal fixed value of all the randomness.
\end{fact}

Next, we define the parallel repetition of a $k$-player game, which corresponds to playing $n$ independent copies of the game in parallel.

\begin{definition} (Parallel Repetition of a game)
	Let $\mc G = (\mc X, \mc A, \mu, Q, V)$ be a $k$-player game. 
	We define its $n$-fold repetition to be the game $\mc G^{n} = (\mc X^{n}, \mc A^{n}, \mu^{n}, Q^{n}, V^{n})$ where:
	\begin{enumerate}
		\item The sets $\mc X^{n}, \mc A^{n}, Q^{n}$ are the $n$-fold product of the sets $\mc X, \mc A, Q$ respectively.
		\item The distribution $\mu^{n}$ is the $n$-fold product of the distribution $\mu$, satisfying $\mu^{n}(x) = \prod_{i=1}^n \mu(x_i)$ for each $x\in \mc X^{n}$.
		\item The predicate $V^{n}:\mc X^{n}\times \mc A^{n}\to \set{0,1}$ is defined as $V^{n}(x, a) = \prod_{i=1}^n V(x_i, a_i)$.
	\end{enumerate}
\end{definition}

\textbf{Notation:} we use subscripts to denote the coordinates in the parallel repetition, and superscripts to denote the players. 
That is, for any $x\in \mc X^{n}$, and subsets $S\subseteq [n], T\subseteq [k]$, we shall use $x_S^\up{T}$ to denote the questions in coordinates $S$, that the players in set $T$ receive.
For example, for $i\in[n]$ and $j\in[k]$, we will use $x_i^\up{j}$ to refer to the question to the $j\textsuperscript{th}$ player in the $i\textsuperscript{th}$ repetition of the game.
Similarly, $x_i$ will refer to the vector of questions to the $k$ players in the $i\textsuperscript{th}$ repetition, and $x^\up{j}$ will refer to the vector of questions received by the $j\textsuperscript{th}$ player over all repetitions.


\section{Forbidden Subgraphs}\label{sec:for_sub}

In this section, we describe the results of~\cite{FV02}, suitably generalized to games with more than 2 players.

Let $\mc G = (\mc X, \mc A, \mu, Q, V)$ be a $k$-player game.
The forbidden subgraph technique is a method to upper bound the value $\val(\mc G^{n})$ in terms of the combinatorial properties of the $k$-partite \emph{game hypergraph} corresponding to the set $Q$, defined as follows:

\begin{definition}(Game Hypergraph)\label{defn:game_hypergraph}
	Given a product set $\mc X = \mc X^\up{1}\times \dots \times \mc X^\up{k}$, and any set $Q\subseteq \mc X$, we associate with $Q$ the natural $k$-partite hypergraph and denote it by $H_{\mc X, Q}$.
	This has vertex set $\mc X^\up{1}\cup \dots \cup \mc X^\up{k}$, and edge set $Q$; that is, every $(x^\up{1},\dots,x^\up{k})\in Q$ forms an edge in $H_{\mc X, Q}$.
	
	We also define the (undirected) $k$-partite graph $F_{\mc X,Q}$ as the graph obtained by projecting the hypergraph $H_{\mc X, Q}$ to a graph.
	This has vertex set $\mc X^\up{1}\cup \dots \cup \mc X^\up{k}$, and contains the edges $\set{x^\up{i}, x^\up{j}}$ for each $(x^\up{1},\dots,x^\up{k})\in Q$ and each $i,j\in [k],i\not=j$.
	
	For a $k$-player game $\mc G = (\mc X, \mc A,\mu, Q, V)$, we call $H_{\mc X, Q}$ as the game hypergraph.
\end{definition}

Next, we define the central notion of a forbidden subgraph.

\begin{definition}\label{defn:forbidden_subgraph}(Forbidden Subgraph)
Let $\mc X = \mc X^\up{1}\times \dots \times \mc X^{\up{k}}$ be a product set, let $Q\subseteq \mc X$ be of size $\abs{Q}=q$, and let $n\in \N$.

A hypergraph $H\subseteq H_{\mc X^n,Q^{n}}$ is called a forbidden subgraph, if there is a coordinate $i\in [n]$ such that:
\begin{enumerate}
	\item The hypergraph $H$ is a union of $q$ edges $e(1), \dots, e(q)$ of the hypergraph $H_{\mc X^n,Q^{n}}$.
	
		Recall that each hyperedge $e(r)\in Q^n$, for $r\in [q]$, contains the questions to all $k$ players for each of the $n$ coordinates/repetitions of the game $\mc G$.
	
	\item The set of projections $\set{e(1)_i, \dots e(q)_i}$ equals $Q$.
		
		Here, for $r\in [q], i\in [n]$, we denote by $e(r)_i$ the vector of questions to the $k$ players in the $i$\textsuperscript{th} coordinate of the game, in the edge $e(r)$.
	
		Note that this implies, in particular, that the hypergraph $H$ is isomorphic to $H_{\mc X, Q}$.
	\item For each $j\in [k]$, and $r,r'\in [q]$, we have that $e(r)_i^\up{j} = e(r')_i^\up{j} \implies e(r)^\up{j} = e(r')^\up{j}$.
	
		In words, for each player $j\in [k]$, among the $q$ edges $e(1), \dots, e(q)$, the question to the $j$\textsuperscript{th} player in coordinate $i$ determines the entire vector of questions to the $j$\textsuperscript{th} player.
\end{enumerate}

We define the function $E_Q(n)$ to be the maximum density $\frac{\abs{W}}{q^n}$ of a subset $W\subseteq Q^{ n}$, such that the hypergraph $H_{\mc X^n,W}$ contains no forbidden subgraph.
\end{definition}

With the above definition in hand, we have the following proposition (justifying the term forbidden subgraph):

\begin{proposition}\label{prop:game_val_ub_forbidden}
	Let $\mc G = (\mc X, \mc A, \mu, Q, V)$ be a $k$-player game with value $\val(\mc G) < 1$, and let $n \in \N$.
	Consider an arbitrary strategy for the players for the game $\mc G^{n}$ (as in Definition~\ref{defn:multiplayer_games}), and let the set $W \subseteq  Q^{n}$ be the set of inputs on which this strategy wins.
	Then, $H_{\mc X^n,W}$ has no forbidden subgraph.
\end{proposition}
\begin{proof}
	Fix some strategies $\brac{f^\up{j}:(\mc X^\up{j})^{n}\to (\mc A^\up{j})^{n}}_{j\in[k]}$ for the players, and suppose, for the sake of contradiction, that $H_{\mc X^n,W}$ contains some forbidden subgraph $H\subseteq H_{\mc X^n,W}\subseteq H_{\mc X^n,Q^{n}}$.
	Let $e(1),\dots,e(q) \in W$ denote the edges of $H$, and let $i\in [n]$ be as in Definition~\ref{defn:forbidden_subgraph}.
	Using the forbidden subgraph $H$, we define a strategy for the game $\mc G$ that wins with value 1, and this will be a contradiction.
	
	The strategy for each player $j\in [k]$ is defined as follows: On input $x^\up{j}\in \mc X^\up{j}$, the player chooses an arbitrary edge $e(r)$, for $r\in [q]$, having question $x^\up{j}$ for player $j$ in coordinate $i$; that is, $e(r)^\up{j}_i = x^\up{j}$.
	Then, the player outputs the $i$\textsuperscript{th} coordinate of $f^\up{j}(e(r)^\up{j})$.
	
	First, observe that the above strategy is valid: by the second point in Definition~\ref{defn:forbidden_subgraph}, for every possible question $x^\up{j}$, asked to player $j$ with non-zero probability, there is some edge $e(r)$ with $e(r)^\up{j}_i = x^\up{j}$.
	Further, by the third point, each such edge has the same vector of questions $e(r)^\up{j} \in (\mc X^\up{j})^{n}$ for the $j$\textsuperscript{th} player, ensuring that the strategy does not depend on the arbitrary choice of the edge $e(r)$.
	
	Next, we show that this strategy wins the game $\mc G$ with probability 1. 
	For all questions $x=(x^\up{1},\dots,x^\up{k})\in Q$ to the $k$ players, by the second point in Definition~\ref{defn:forbidden_subgraph}, there exists an edge $e(r), r\in [q]$ such that $e(r)_i = x$.
	Now, since the players win on the edge $e(r)$ (by the definition of the winning set $W$), the above strategy also wins on input $x$.
\end{proof}

The above proposition immediately gives us:

\begin{corollary}\label{corr:val_ub_forbidden}
	Let $\mc G = (\mc X, \mc A, \mu, Q, V)$ be a $k$-player game with value $\val(\mc G) < 1$, and with $\mu$ the uniform distribution over $Q$.
	Then, $\val\brac{\mc G^{n}} \leq E_Q(n)$.
\end{corollary}

\begin{remark}\label{remk:unif_dist}
If $\mu$ is not the uniform distribution, it still holds that $\val\brac{\mc G^{n}} \leq 2\cdot E_Q(\lfloor\lambda n \rfloor)$, for some constant $\lambda>0$.

Roughly speaking, we can write $\mu = \epsilon \cdot \pi + (1-\epsilon)\cdot \mu'$, where $\pi$ is the uniform distribution over $Q$, and $\epsilon>0$.
Then, to do well on $n$ independent copies of $\mc G$, one must do well on at least linearly many copies ($\epsilon n$ in expectation) of the corresponding game with the uniform distribution (see Lemma 3.14 in~\cite{GHMRZ22}).
\end{remark}

Next, we show that the above result is tight, when the game $\mc G$ is allowed to have large (depending on $n$) answer sets.

\begin{theorem}\label{thm:game_val_lb_forbidden}
	Let $\mc X = \mc X^\up{1}\times \dots \times \mc X^\up{k}$ be a product set, and let $Q\subseteq \mc X$ be such that the graph $F_{\mc X,Q}$ is connected (see Definition~\ref{defn:forbidden_subgraph}).
	
	Then, for every $n\in \N$, there exists a game $\mc G = (\mc X, \mc A, \mu, Q, V)$, with $\mu$ the uniform distribution on $Q$, and $\abs{\mc A} = n^k\cdot \abs{\mc X}^n $, such that $\val\brac{\mc G^{n}} = E_Q(n)$.
\end{theorem}
\begin{proof}
	Let $\mc X, Q$ be as in the statement of the theorem, and let $n\in \N$.
	Let $W\subseteq Q^{n}$ be a set of size $E_Q(n)\cdot \abs{Q}^n$, such that $H_{\mc X^n,W}$ contains no forbidden subgraph. 
	
	We define a game $\mc G = (\mc X, \mc A, \mu, Q, V)$ as follows:
	\begin{enumerate}
		\item The distribution $\mu$ is uniform over $Q$.
		\item For each $j\in [k]$, the answer set $\mc A^\up{j} = [n]\times (\mc X^\up{j})^{n}$.
			That is, any answer $a^\up{j}$ will be a tuple $(i^\up{j}, y^\up{j})$, for $i^\up{j}\in [n], y^\up{j}\in (\mc X^\up{j})^{n}$.
		\item The predicate $V$, on questions $(x^\up{1},\dots,x^\up{k})\in \mc X$, and answers $\ (i^\up{1}, y^\up{1}),\dots, (i^\up{k}, y^\up{k})$, accepts if and only if $(y^\up{1},\dots,y^\up{k})\in W$, and $i^\up{1}=i^\up{2}=\dots =i^\up{k}=i$ for some $i\in [n]$, and $(y^\up{1}_i,\dots,y^\up{k}_i) = (x^\up{1},\dots,x^\up{k})$.
	\end{enumerate}
	
	This game satisfies:
	\begin{enumerate}
		\item $\val\brac{\mc G^{n}} \geq E_Q(n)$: Consider the following strategy for the game $\mc G^{n}$.
		Each player $j\in [k]$, on input $x^\up{j}\in (\mc X^\up{j})^{n}$, outputs for each coordinate $i\in [n]$, the answer $(i,x^\up{j})$.
		By the definition of the game $\mc G$, this strategy wins on all inputs $x\in W$.
		Hence, $\val\brac{\mc G^{n}} \geq \frac{\abs{W}}{\abs{Q}^n} = E_Q(n)$.
		
		\item $\val\brac{\mc G^{n}} \leq E_Q(n)$: By Proposition~\ref{prop:game_val_ub_forbidden}, it suffices to show that $\val(\mc G)<1$.
		Suppose, for the sake of contradiction, that $\val(\mc G) = 1$, and fix a strategy for the players that achieves this.
		
		Observe that on any input $(x^\up{1},\dots,x^\up{k})\in \mc X$, the predicate $V$ requires that each player must output the same index $i\in [n]$.
		Since the graph $F_{\mc X,Q}$ (see Definition~\ref{defn:game_hypergraph}) is connected, this implies that each player outputs the same $i$ on all inputs (that occur in some question in $Q$). 
		
		Now, for any input $(x^\up{1},\dots,x^\up{k})\in Q$ we know the players answer $(i,y^\up{1}),\dots,(i,y^\up{k})$, with $(y^\up{1},\dots,y^\up{k})\in W\subseteq Q^{n}$ such that $(y^\up{1}_i,\dots,y^\up{k}_i) = (x^\up{1},\dots,x^\up{k})$.
		Then, the set of all such $(y^\up{1},\dots,y^\up{k})$, obtained by going over all possible $(x^\up{1},\dots,x^\up{k})\in Q$, forms a forbidden subgraph contained in hypergraph $H_{\mc X^n,W}$.
		This is a contradiction.
		\qedhere
	\end{enumerate}
\end{proof}

\begin{remark}
	The requirement that $F_{\mc X,Q}$ be connected is only for simplicity.
	Any game $\mc G$ with question set $Q$ is essentially a \emph{disjoint union} of different games, corresponding to the connected components of the graph $F_{\mc X,Q}$.
	Hence, modifying the definition of forbidden subgraphs to allow for different coordinates $i\in [n]$ for each such component, we can obtain the above theorem without the connectivity assumption (and still preserve Proposition~\ref{prop:game_val_ub_forbidden} and Corollary~\ref{corr:val_ub_forbidden}).
\end{remark}

\begin{remark}\label{remk:inf_length}
	In Theorem~\ref{thm:game_val_lb_forbidden}, if we allow a countably-infinite answer set $\mc A$, we can construct a single game $\mc G$ with $\val(\mc G^{n}) = E_Q(n)$ for all $n\in \N$ (note that Corollary~\ref{corr:val_ub_forbidden} holds even when $\mc A$ is infinite).
	Roughly, this is done as follows:
	
	Let $W_1,W_2,\dots$ be sets such that for each $n\in \N$, the set $W_n \subseteq Q^{n}$ has no forbidden subgraph and is of size $E_Q(n)\cdot \abs{Q}^n$.
	In the game $\mc G$, each player $j\in [k]$ will output $n^\up{j}\in \N, i^\up{j}\in [n^\up{j}], y^\up{j}\in (\mc X^\up{j})^{n}$.
	The predicate $V$, on questions $x=(x^\up{1},\dots,x^\up{k})$, and these answers, accepts if and only if the players output the same $n\in \N$ and the same $i\in [n]$, and such that $(y^\up{1},\dots,y^\up{k}) \in W_n$ and $(y^\up{1}_i,\dots,y^\up{k}_i)=x$.
\end{remark}

\section{Connections to Extremal Combinatorics}

In this section, we show several connections to extremal/additive combinatorics.

\subsection{Density Hales-Jewett Theorem}\label{sec:dhj_conn}

We recall Definition~\ref{defn:comb_line}:
\begin{definition}
\label{defn:comb_line_main}
(Combinatorial lines)
	Let $q, n\in \N$.
	A collection of $q$ points $a(1),\dots, a(q) \in [q]^n$ is said to be a combinatorial line if:
	\begin{enumerate}
		\item For all $i\in [n]$, either $(a(1)_i,\dots, a(q)_i)= (1,2,3,\dots,q)$ or $a(1)_i=a(2)_i=\dots=a(q)_i$.
		\item For some $i\in [n]$, it holds that $(a(1)_i,\dots, a(q)_i) = (1,2,3,\dots,q)$.
	\end{enumerate}
	We define $r_q(n) = \frac{\abs{A}}{q^n}$, where $A\subseteq [q]^n$ is a largest set  containing no combinatorial line.
\end{definition}

Verbitsky~\cite{Ver96} proved that for any $2$-player game $\mc G = (\mc X, \mc A, \mu, Q, V)$ with $\mu$ the uniform distribution over $Q$, and with value less than 1, it holds that $\val(\mc G^{n}) \leq r_{\abs{Q}}(n)$.
This result easily extends to all $k$-player games, and gives the current best decay bound for the value of parallel repetition of general $k$-player games with value less than 1.
As shown by~\cite{FV02}, the above result is in fact a special case of forbidden subgraph bounds, via Corollary~\ref{corr:val_ub_forbidden}, as follows:

\begin{proposition}\label{prop:comb_lin_forb}
	Let $\mc X = \mc X^\up{1}\times \dots \times \mc X^\up{k}$ be a product set, and let $Q\subseteq \mc X$ be of size $\abs{Q}=q$.
	Then, for every $n\in \N$, it holds that $E_Q(n)\leq r_q(n)$.
\end{proposition}
\begin{proof}
	By identifying $Q$ with $[q]$, it follows easily from Definition~\ref{defn:forbidden_subgraph} and Definition~\ref{defn:comb_line_main} that any combinatorial line in $[q]^n$ corresponds to a forbidden subgraph in $H_{\mc X^n,Q^{n}}$.
\end{proof}

H{\k a}z{\l}a, Holenstein and Rao~\cite{HHR16} show a converse to the above result, by presenting games (one for each $n$) that achieve equality in the above proposition:

\begin{proposition}
	Let $q\geq 3$ be a positive integer, $\mc X = \set{0,1}^q$, and let \[Q = \set{(1,0,0,\dots,0), (0,1,0,\dots,0), \dots,(0,0,\dots,0,1)} \subseteq \set{0,1}^q\] be the set of size $q$ containing all vectors with 1 one and $(q-1)$ zeros.
	
	Then, for every $n\in \N$, it holds that $E_{Q}(n) = r_q(n)$.
\end{proposition}
\begin{proof}
	Let $q\geq 3$ be a positive integer, and let $n\in \N$.
	We will show that any forbidden subgraph of $H_{\mc X^n,Q^{n}}$ corresponds to a combinatorial line.
	Combined with Proposition~\ref{prop:comb_lin_forb}, this implies the desired result.
	
	Consider any forbidden subgraph in $H_{\mc X^n,Q^{n}}$, given by edges $e(1),\dots,e(q)$, and let $i\in [n]$ be the coordinate as in Definition~\ref{defn:forbidden_subgraph}.
	Identifying $[q]$ with $Q$ (say for example under the map $r\mapsto (0,\dots,0,1,0,\dots,0)$ with 1 in the $r$\textsuperscript{th} position, for all $r\in [q]$), we can view each $e(1),\dots,e(q)$ as an element of $[q]^n$.
	Next, we show that these form a combinatorial line.
	
	By the second and third points in Definition~\ref{defn:forbidden_subgraph}, we can define for each $j\in [q]$ questions $x^\up{j},y^\up{j} \in (\mc X^\up{j})^{n}$, such that $x^\up{j}_i= 0,\ y^\up{j}_i=1$; the questions $x^\up{j}$ (resp. $y^\up{j}$) will correspond to the input of the $j$\textsuperscript{th} player, among the edges of the forbidden subgraph, when the input in the $i$\textsuperscript{th} coordinate is $0$ (resp. 1).
	Then, after reordering if needed, we have that for each $r\in [q]$, $e(r) = \brac{x^\up{1},\dots,x^\up{r-1}, y^\up{r},x^\up{r+1}, \dots,x^\up{q}} \in Q^{n}$.
	
	With the above, we check that $e(1),\dots,e(q)$ form a combinatorial line.
	For any coordinate $i'\in [n]$ consider the projection of these edges on coordinate $i'$, given by $e(1)_{i'},\dots,e(q)_{i'}$:
	\begin{enumerate}
		\item In coordinate $i'=i$, we have that $(e(1)_{i},\dots,e(q)_{i})$ corresponds to the tuple $(1,2,\dots,q)$, under the mapping between $Q$ and $[q]$.
		\item Suppose $i'\not=i$. For each $j\in [k]$, let $a^\up{j} = x^\up{j}_{i'} \in \set{0,1}$ and $b^\up{j} = y^\up{j}_{i'} \in \set{0,1}$.
		Then, for each $r\in [k]$, we have $e(r)_{i'} = (a^\up{1},\dots,a^\up{r-1},b^\up{r},a^\up{r+1},\dots,a^\up{q}) \in Q\subseteq \set{0,1}^k$.
		Consider the following cases:
		\begin{itemize}
			\item Suppose $a^\up{r} = 1$ for some $r\in [q]$. By symmetry, we may assume $a^\up{1}=1$.
			Using that $e(r)_{i'}\in Q$ for each $r>1$, we get that $a^\up{2}=b^\up{2}=\dots=a^\up{q}=b^\up{q}=0$; note that this step uses $q\geq 3$.
			Now, using the same for $r=1$, we get that $b^\up{1}=1$.
			Hence, for each $r\in [q]$, we have $e(r)_{i'} = (1,0,0,\dots,0)$.
			
			\item If $a^\up{r}=0$ for each $r\in [q]$, using that $e(r)_{i'}\in Q$, we get $b^\up{r}=1$, for each $r\in [q]$.
				In this case, $(e(1)_{i'},\dots,e(k)_{i'})$ corresponds to the tuple $(1,2,\dots,q)$, under the mapping between $Q$ and $[q]$.
			\qedhere
		\end{itemize}
	\end{enumerate}
\end{proof}

For the set $Q$ as in the above proposition, it is easily checked that the graph $F_{\mc X,Q}$ (see Definition~\ref{defn:forbidden_subgraph}) is connected.
Then, by Theorem~\ref{thm:game_val_lb_forbidden}, we get the following corollary.

\begin{corollary}\label{corr:hj_inv}
	Let $q\geq 3$ be a positive integer, $\mc X = \set{0,1}^q$, and let \[Q = \set{(1,0,0,\dots,0), (0,1,0,\dots,0), \dots,(0,0,\dots,0,1)} \subseteq \set{0,1}^q\] be the set of size $q$ containing all vectors with 1 one and $(q-1)$ zeros.
	Then, for every $n\in \N$, there exists a $q$-player game $\mc G = (\mc X, \mc A, \mu, Q, V)$, with $\mu$ the uniform distribution on $Q$, and such that $\val\brac{\mc G^{n}} = r_q(n)$.
\end{corollary}

\begin{remark}
	We note that for $q=2$, as noted earlier, we have $r_2(n) = \Theta\brac{\frac{1}{\sqrt{n}}}$.
	However, for any $k\in \N$, $\mc X = \mc X^\up{1}\times \dots \times \mc X^\up{k}$, and $Q \subseteq \mc X$ of size $\norm{Q}=2$, it is easily seen that $E_Q(n) \leq \frac{1}{2^n}$, since any two distinct inputs in $Q^n$ form a forbidden subgraph.
	
	For $\norm{Q}=1$ and $q=1$, we always have $E_Q(n)=r_1(n)=0$ trivially. 
\end{remark}


\subsection{Square-Free sets over \texorpdfstring{$\F_2^n$}{F\_2\^{}n}}\label{sec:square_conn}

We recall Definition~\ref{defn:corner_sq}:

\begin{definition} (Square free sets)
	Let $n\in \N$.
	A square in $\F_2^n\times \F_2^n$ is a set of the form \[\set{(x,y), (x+d,y), (x,y+d), (x+d,y+d)},\] where $x,y,d\in  \F_2^n,\ d\not=0$.
			
			Define $r_{\square}(\F_2^n)$ to be the maximum density $\frac{\abs{A}}{4^n}$ of a set $A\subseteq \F_2^n\times \F_2^n$ containing no square.
\end{definition}

We show that the question set of the GHZ game exactly captures this problem.

\begin{proposition}
	Let $\mc X = \F_2^3$, and let $Q = \set{(x,y,z)\in \F_2^3 : x+y+z=0} \subseteq \mc X$.
	Then, for every $n\in \N$, $E_Q(n) = r_{\square}(\F_2^n)$.
\end{proposition}
\begin{proof}
	Consider any $n\in \N$.
	
	First, we show that $E_Q(n) \leq r_{\square}(\F_2^n)$.
	Let $W\subseteq Q^{n}$ be such that $H_{\mc X^n,W}$ has no forbidden-subgraph.
	It suffices to show that $A =\set{(x,y):(x,y,x+y)\in W} \subseteq \F_2^n\times \F_2^n$ is a square free set, since $\frac{\abs{A}}{4^n} = \frac{\abs{W}}{\abs{Q}^n}$.
	Suppose, for the sake of contradiction, that $A$ contains a square.
	Then, there exists $x,y,d\in \F_2^n, d\not=0$ such that \[\set{(x,y,x+y), (x+d,y,x+y+d), (x,y+d,x+y+d), (x+d,y+d,x+y)}\subseteq W.\]
	This, by Definition~\ref{defn:forbidden_subgraph}, is a forbidden subgraph, choosing $i\in [n]$ to be any coordinate with $d_i\not=0$.
	
	Next, we show that $E_Q(n) \geq r_{\square}(\F_2^n)$.
	Let $A \subseteq \F_2^n\times \F_2^n$ be a square-free set.
	It suffices to show that $W = \set{(x,y,x+y):(x,y)\in A}\subseteq Q^{n}$ is such that $H_{\mc X^n,W}$ has no forbidden subgraph.
	Suppose, for the sake of contradiction, that $H_{\mc X^n,W}$ contains a forbidden subgraph.
	By the second and third points in Definition~\ref{defn:forbidden_subgraph}, we can find $x,y,z, x',y',z'\in \F_2^n$, and $i\in [n]$, such that $x_i=y_i=z_i=0$ and $x'_i=y'_i=z'_i=1$, and such that the edges of the forbidden subgraph are \[\set{(x,y,z), (x,y',z'), (x',y,z'), (x',y',z)} \subseteq W.\]
	Then, we have $x+y+z = x+y'+z' = x'+y+z'= x'+y'+z=0$.
	This implies
	\[ x+x'=y+y'=z+z' := d,\]
	with $d\not=0$ (since $d_i=1$), and such that $z=x+y$.
	Substituting this into the above expression, we get
	\[\set{(x,y,x+y), (x+d,y,x+y+d), (x,y+d,x+y+d), (x+d,y+d,x+y)}\subseteq W,\]
	or, equivalently,
	\[\set{(x,y), (x+d,y), (x,y+d), (x+d,y+d)}\subseteq A,\]
	which is a contradiction.
\end{proof}

For the set $Q$ as in the above proposition, it is easily checked that the graph $F_{\mc X,Q}$ (see Definition~\ref{defn:forbidden_subgraph}) is the complete $3$-partite graph, and hence connected.
Then, by Theorem~\ref{thm:game_val_lb_forbidden}, we get the following corollary:

\begin{corollary}\label{corr:square}
	Let $\mc X = \F_2^3$, and let $Q = \set{(x,y,z)\in \F_2^3 : x+y+z=0} \subseteq \mc X$.
	Then, for every $n\in \N$, there exists a $3$-player game $\mc G = (\mc X, \mc A, \mu, Q, V)$, with $\mu$ the uniform distribution on $Q$, and such that $\val(\mc G^{n}) = r_{\square}(\F_2^n)$.
\end{corollary}

\begin{remark}\label{rmk:square_in_par_rep}
	In the above proof, we see that configurations of the form \[\set{(x,y,x+y), (x+d,y,x+y+d), (x,y+d,x+y+d), (x+d,y+d,x+y)},\]
	with $x,y,d\in \F_2^n, d\not=0$
	are forbidden subgraphs for the question set $Q^{n}$ of the $n$-fold GHZ game.
	This was first observed by~\cite{GHMRZ21}, who called these configurations as bow-ties, and used these to prove polynomial decay bounds for parallel repetition of all games with question set as $Q$, and with value less than 1 (and with bounded answer lengths).
\end{remark}


\subsection{Grid-Free sets over \texorpdfstring{$\F^n$}{F\^{}n}}\label{sec:grids}

Throughout this section, we shall work over a fixed finite field $\F$.
We show that the results of the previous section extend to the more general setting of sets free of $k$-dimensional grids over $\F^{nk}$.
We recall Definition~\ref{defn:grid_free}:

\begin{definition}(Grid free sets)
Let $k,n\in \N$.
	
	A grid in $(\F^n)^k$ is a set of the form $\set{x + \alpha\otimes d : \alpha\in \F^k}$ for some $x\in (\F^n)^k$,\ $d\in \F^n,\ d\not=0$.
	We use the Kronecker product notation $\alpha\otimes d = (\alpha^\up{1}d,\dots,\alpha^\up{k}d) \in (\F^{n})^k,$ for $\alpha\in \F^k,\ d\in \F^n$.
	
	Define $r_{grid}(\F,k,n)$ as the maximum density $\frac{\abs{A}}{\abs{\F}^{kn}}$ of a set $A\subseteq (\F^n)^k$ containing no grid.
\end{definition}
Observe that any combinatorial line in $[\abs{\F}^k]^n$ corresponds to a grid in $\F^{kn}$, and hence $r_{grid}(\F,k,n) \leq r_{\abs{\F}^k}(n) = o_n(1)$.

Next, we show that grid-free sets correspond to parallel repetition of games with query distribution as linear/affine subspaces.
Since grids, in some sense, represent the largest linear forbidden structures in $(\F^n)^k$,  understanding parallel repetition may give a way to attack this large class of problems in extremal/additive combinatorics.

\begin{proposition}\label{prop:lin_game_grid_free}
	Let $k,\ell\in \N$, $\mc X = \F^{\ell}$, and let $Q\subseteq \F^\ell$ be an affine subspace of dimension $k \leq \ell$.
	Then, $E_Q(n) \leq r_{grid}(\F,k,n)$.
\end{proposition}
\begin{proof}
	Let $\varphi:\F^k\to \F^{\ell}$ be an injective affine map such that $Q = \set{\varphi(x):x\in \F^k}$.
	
	Consider any $n\in \N$, and let $W\subseteq Q^{n}$ be such that $H_{\mc X^n,W}$ has no forbidden-subgraph.
	Let $A = \set{x: \varphi^n(x)\in W} \subseteq (\F^n)^k$; here $\varphi^n$ is the affine map acting independently on each of the $n$ coordinates, that is, $\varphi^n(y) = (\varphi(y_1),\dots,\varphi(y_n))$ for any $y\in (\F^k)^n$.	
	 It suffices to show that is a grid-free set, since $\frac{\abs{A}}{\abs{\F}^{kn}} = \frac{\abs{W}}{\abs{Q}^n}$.
	 
	 Suppose, for the sake of contradiction, that $A$ contains a grid.
	Then, there exists $x\in (\F^n)^k$ and $d\in \F^n, d\not=0$ such that $\varphi^n(x+\alpha\otimes d) \in W$ for each $\alpha\in \F^k$.
	This, by Definition~\ref{defn:forbidden_subgraph}, is a forbidden subgraph, choosing $i\in [n]$ to be any coordinate with $d_i\not=0$, as follows:
	\begin{enumerate}
		\item The $\abs{\F^k}$ edges are $\set{\varphi^n(x+\alpha\otimes d) : \alpha\in \F^k}$.
		\item As $d_i\not=0$, we have that  $\set{x_i+\alpha\cdot d_i: \alpha\in \F^k} = \F^k$. This gives us \[\set{\varphi^n(x+\alpha\otimes d)_i : \alpha\in \F^k}=\set{\varphi(x_i+\alpha\cdot d_i) : \alpha\in \F^k} = Q.\]
		\item Fix any $j\in [\ell]$, and let $\psi:\F^k\to \F$ be the linear map denoting the $j$\textsuperscript{th} coordinate of $\varphi$.
		Now, suppose $\alpha,\beta\in \F^k$ are such that $\psi^n(x+\alpha\otimes d)_i = \psi^n(x+\beta\otimes d)_i$.
		Then, we have $\psi(x_i)+\psi(\alpha)\cdot d_i-\psi(0) = \psi(x_i)+\psi(b)\cdot d_i-\psi(0)$, which implies $\psi(a) = \psi(\beta)$.
		Now, for each $i'\in [n]$, we have $\psi^n(x+a\otimes d)_{i'} =\psi(x_{i'})+\psi(\alpha)\cdot d_{i'}-\psi(0) = \psi(x_{i'})+\psi(b)\cdot d_{i'}-\psi(0)= \psi^n(x+b\otimes d)_{i'}$.
		\qedhere
	\end{enumerate}
\end{proof}

We show that appropriately chosen $Q$ can achieve equality in the above proposition:

\begin{proposition}
	Suppose that  $\abs{\F} = p^r$ for a prime $p$, and $r\in \N$.
	Then, for every integer $k\geq 2$, $\mc X = \F^{k+r}$, there exists a linear subspace $Q\subseteq \mc X$ of dimension $k$, such that for every $n\in \N$, $E_Q(n) = r_{grid}(\F,k,n)$.
\end{proposition}
\begin{proof}
	Let $k\geq 2$ be an integer and $\mc X = \F^{k+r}$.
	Let $G \subseteq \F$ be a set of size $r$ that generates the additive group over $\F$.
	The set $Q\subseteq \mc X$ is defined as
	\[ Q = \set{(x,y): x\in \F^k, y = \brac{t\cdot x^\up{1}+x^\up{2}+\dots + x^\up{k}}_{t\in G}} .\]

	Consider any $n\in \N$.
	The inequality $E_Q(n) \leq  r_{grid}(\F,k,n)$ follows from Proposition~\ref{prop:lin_game_grid_free}.
	We show that $E_Q(n) \geq r_{grid}(\F,k,n)$.
	Let $A\subseteq (\F^n)^k$ be a grid-free set; it suffices to show that $W = \set{ (x,y): x\in A, y = (t\cdot x^\up{1}+x^\up{2}+x^\up{3}+\dots+ x^\up{k})_{t \in G} }  \subseteq Q^n$ is such that $H_{\mc X^n,W}$ has no forbidden subgraph.
	Suppose, for the sake of contradiction, that $H_{\mc X^n,W}$ contains a forbidden subgraph.
	
	By the second and third points in Definition~\ref{defn:forbidden_subgraph}, we can find inputs $\set{x^\up{j}(\omega) \in \F^n}_{j\in [k], \omega\in \F}$ and $\set{y^\up{t}(\omega) \in \F^n}_{t\in G, \omega\in \F}$ for the $k+r$ players, and $i\in [n]$, such that
	\begin{enumerate}
		\item $x^\up{j}(\omega)_i = \omega$ for each $j\in [k],\omega\in \F$ and $y^\up{t}(\omega)_i = \omega$ for each $t\in G,\omega\in \F$.
		\item For each $\alpha\in \F^k$, it holds that \[\brac{\brac{x^\up{j}(\alpha^\up{j})}_{j\in [k]}, \brac{y^\up{t}(t\cdot \alpha^\up{1}+\alpha^\up{2}+\dots + \alpha^\up{k})}_{t\in G}} \in W\subseteq Q^n.\]
		In particular, for each $\alpha,\beta\in \F^k$ and $t\in G$ satisfying $t\cdot \alpha^\up{1}+\sum_{j=2}^k\alpha^\up{j} = t\cdot \beta^\up{1}+\sum_{j=2}^k\beta^\up{j}$, it holds that $t\cdot x^\up{1}(\alpha^\up{1})+\sum_{j=2}^kx^\up{j}(\alpha^\up{j}) = t\cdot x^\up{1}(\beta^\up{1})+\sum_{j=2}^kx^\up{j}(\beta^\up{j})$, as each of these must equal $y^\up{t}\brac{t\cdot \alpha^\up{1}+\sum_{j=2}^k \alpha^\up{j}}=y^\up{t}\brac{t\cdot \beta^\up{1}+\sum_{j=2}^k \beta^\up{j}}$.
	\end{enumerate}
	Now, by the above, we get (since $k\geq 2$):
	\begin{enumerate}
		\item For each $j\in [k],\omega\in \F$, it holds that $\omega+0+\dots+0 = 0+\dots+0+\omega+0+\dots+0$ with $\omega$ in the $j$\textsuperscript{th} position, and hence $x^\up{1}(\omega)+x^\up{j}(0) = x^\up{1}(0)+x^\up{j}(\omega)$, and hence  $x^\up{j}(\omega)- x^\up{j}(0) = x^\up{1}(\omega)-x^\up{1}(0)$.
		\item For each $\omega,\omega'\in \F$, by a similar reasoning, we get that
		\begin{align*}
			x^\up{1}(\omega+\omega') - x^\up{1}(0) &= x^\up{1}(\omega+\omega') + x^\up{2}(0)-x^\up{1}(0)-x^\up{2}(0)
			\\&= x^\up{1}(\omega)+ x^\up{2}(\omega')-x^\up{1}(0)-x^\up{2}(0)
			\\&= \brac{x^\up{1}(\omega)-x^\up{1}(0)}+ \brac{x^\up{2}(\omega')-x^\up{2}(0)}
			\\&= \brac{x^\up{1}(\omega)-x^\up{1}(0)} + \brac{x^\up{1}(\omega')-x^\up{1}(0)}.
		\end{align*}
		
		\item For each $t\in G$, we have $t\cdot 1 +0+\dots+ 0 = t\cdot 0 + t + 0+\dots+0$, and hence $t\cdot x^\up{1}(1)+x^\up{2}(0) = t\cdot x^\up{1}(0)+x^\up{2}(t)$.
		Combining with the above, we get that for each $t\in G$,
		\[ x^\up{1}(t)-x^\up{1}(0) = x^\up{2}(t)-x^\up{2}(0) = t\cdot (x^\up{1}(1)-x^\up{1}(0)). \]
	\end{enumerate}
	Since $G$ generates the additive group over $\F$, the above gives us that for each $j\in [k],\omega\in \F$, after writing $\omega = t_1+\dots +t_s$ for $t_1,\dots, t_s\in G$,
	\begin{align*}
		x^\up{j}(\omega)- x^\up{j}(0) &= x^\up{1}(\omega)-x^\up{1}(0)
		\\&= \sum_{i\in [s]} \brac{x^\up{1}(t_i)-x^\up{1}(0)}
		\\&= \sum_{i\in [s]} t_i\cdot \brac{x^\up{1}(1)-x^\up{1}(0)}
		\\&= \omega\cdot \brac{x^\up{1}(1)-x^\up{1}(0)}.
	\end{align*}
	
	Hence, for $z = (x^\up{1}(0), \dots, x^\up{k}(0)) \in (\F^n)^k$ and $d = x^\up{1}(1)-x^\up{1}(0)\in \F^n$, we have that for each $\alpha\in \F^k$, $(x^\up{1}(\alpha^\up{1}), \dots, x^\up{k}(\alpha^{\up{k}})) = z + \alpha\otimes d \in A$.
	This contradicts that $A$ is grid-free.
\end{proof}

For the set $Q$ defined in the proof of the above proposition, it is easily checked that the graph $F_{\mc X,Q}$ (see Definition~\ref{defn:forbidden_subgraph}) is the complete $(k+r)$-partite graph, and hence connected.
Then, by Theorem~\ref{thm:game_val_lb_forbidden}, we get the following corollary:

\begin{corollary}\label{corr:grid_free}
	Suppose that  $\abs{\F} = p^r$ for a prime $p$, and $r\in \N$.
	Then, for every integer $k\geq 2$, $\mc X = \F^{k+r}$, there exists a linear subspace $Q\subseteq \mc X$ of dimension $k$, such that for every $n\in \N$, there exists a $(k+r)$-player game $\mc G = (\mc X, \mc A, \mu, Q, V)$, with $\mu$ the uniform distribution on $Q$, and such that $\val(\mc G^{n}) = r_{grid}(\F,k, n)$.
\end{corollary}

\begin{remark}
	We note that the above proposition cannot be true for $k=1$:
	
	It is easily shown that for any $\ell\in \N$, and affine vector space $Q\subseteq \F^\ell$ of dimension 1, that $E_Q(n) \leq \brac{1-\frac{1}{\abs{\F}}}^n$; roughly, this holds because any set $W\subseteq Q^n$ of size more than this has full projection on some coordinate $i\in [n]$, giving a forbidden subgraph.
	
	On the other hand, for example for $\F = \F_3$, any grid-free set in $\F_3^n$ (with $k=1$) is what is called a cap-set, and such sets of size at least $2.2^n \gg 2^n$ are known to exist~\cite{Edel04}.
	\end{remark}


\section*{Acknowledgements}
We thank Ran Raz for many helpful discussions. Ran politely declined to be a co-author.

\bibliographystyle{alpha}
\bibliography{main.bib}

\end{document}